\documentclass[12pt]{article}
\usepackage{fullpage}
\usepackage{amsmath}
\usepackage{amssymb}
\usepackage{amsthm}
\usepackage{bbm}
\usepackage{color}
\usepackage{hyperref}
\usepackage{graphicx}
\usepackage[margin=1in]{geometry}
\usepackage[norelsize,ruled,vlined,linesnumbered]{algorithm2e}

\theoremstyle{plain}

\newtheorem{theorem}{Theorem}
\newtheorem{lemma}[theorem]{Lemma}

\newtheorem{proposition}[theorem]{Proposition}

\newtheorem{remark}[theorem]{Remark}

\newcommand{\SweepCut}{\ensuremath{\texttt{SweepCut}}\xspace}
\newcommand{\ApxFiedler}{\ensuremath{\texttt{ApxFiedler}}\xspace}

\newcommand{\futuretodo}[1]{}

\title{A Schur Complement Cheeger Inequality}
\author{Aaron Schild \\ University of California, Berkeley EECS \\ \href{mailto:aschild@berkeley.edu}{aschild@berkeley.edu} \footnote{Supported by NSF grant CCF-1816861.}}
\begin{document}

\maketitle

\begin{abstract}
Cheeger's inequality shows that any undirected graph $G$ with minimum nonzero normalized Laplacian eigenvalue $\lambda_G$ has a cut with conductance at most $O(\sqrt{\lambda_G})$. Qualitatively, Cheeger's inequality says that if the relaxation time of a graph is high, there is a cut that certifies this. However, there is a gap in this relationship, as cuts can have conductance as low as $\Theta(\lambda_G)$.

To better approximate the relaxation time of a graph, we consider a more general object. Specifically, instead of bounding the mixing time with cuts, we bound it with cuts in graphs obtained by Schur complementing out vertices from the graph $G$. Combinatorially, these Schur complements describe random walks in $G$ restricted to a subset of its vertices. As a result, all Schur complement cuts have conductance at least $\Omega(\lambda_G)$. We show that unlike with cuts, this inequality is tight up to a constant factor. Specifically, there is a Schur complement cut with conductance at most $O(\lambda_G)$.
 \end{abstract}

\section{Introduction}

For a set of vertices $S$, let $\phi_S$ denote the total weight of edges leaving $S$ divided by the total degree of the vertices in $S$. Throughout the literature, this quantity is often called the conductance of $S$. To avoid confusion with electrical conductance, we call this quantity the \emph{fractional conductance} of $S$. Let $\phi_G$ denote the minimum fractional conductance of any set $S$ with at most half of the volume (total vertex degree). Let $\lambda_G$ denote the minimum nonzero eigenvalue of the normalized Laplacian matrix of $G$. Cheeger's inequality for graphs \cite{A86,AM85} is as follows:

\begin{theorem}[Cheeger's Inequality]\label{thm:cheeger}
For any weighted graph $G$, $\lambda_G/2\le \phi_G\le \sqrt{2\lambda_G}$.
\end{theorem}

Cheeger's inequality was originally introduced in the context of manifolds \cite{C69}. It is a fundamental primitive in graph partitioning \cite{ST14,L07} and for upper bounding the mixing time \footnote{Every reversible Markov chain is a random walk on some weighted undirected graph $G$ with vertex set equal to the state space of the Markov chain. The relaxation time of a reversible Markov chain with transition graph $G$ is defined to be $1/\lambda_G$. This quantity is within a $\Theta(\log(\pi_{\min}))$ factor of the mixing time of the chain, where $\pi_{\min}$ is the minimum nonzero entry in the stationary distribution (Theorems 12.3 and 12.4 of \cite{LPW06}).} of Markov chains \cite{S92}. Motivated by spectral partitioning, much work has been done on higher-order generalizations of Cheeger's inequality \cite{LOT12,LRTV12}. The myriad of applications for Cheeger's inequality and generalizations of it \cite{BSS13,SKM14}, along with the $\Theta(\sqrt{\lambda_G})$ gap between the upper and lower bounds, have led to a long line of work that seeks to improve the quality of the partition found when the spectrum has certain properties (for example, bounded eigenvalue gap \cite{KLLOT13} or when the graph has special structure \cite{KLPT10}.)

Here, we get rid of the $\Theta(\sqrt{\lambda_G})$ gap by taking a different approach. Instead of assuming special combinatorial or spectral structure of the input graph to obtain a tighter relationship between fractional conductance and $\lambda_G$, we introduce a more general object than graph cuts that enables a tighter approximation to $\lambda_G$. Instead of just considering cuts in the given graph $G$, we consider cuts in certain derived graphs of the input graph obtained by Schur complementing the Laplacian matrix of the graph $G$ onto rows and columns corresponding to a subset of $G$'s vertices. Specifically, pick two disjoint sets of vertices $S_1$ and $S_2$, compute the Schur complement of $G$ onto $S_1\cup S_2$, and look at the cut consisting of all edges between $S_1$ and $S_2$ in that Schur complement. Let $\rho_G$ be the minimum fractional conductance of any such cut (defined formally in Section \ref{sec:prelim}). We show that the minimum fractional conductance of any such cut is a constant factor approximation to $\lambda_G$:

\begin{theorem}\label{thm:main}
Let $G$ be a weighted graph. Then

$$\lambda_G/2\le \rho_G\le 25600\lambda_G$$
\end{theorem}

\subsection{Effective Resistance Clustering}

Our result directly implies a clustering result that relates $1/\lambda_G$ to effective resistances between sets of vertices in the graph $G$. Think of the weighted graph $G$ as an electrical network, where each edge represents a conductor with electrical conductance equal to its weight. For two sets of vertices $S_1$ and $S_2$, obtain a graph $H$ by contracting all vertices in $S_1$ to a single vertex $s_1$ and all vertices in $S_2$ to a single vertex $s_2$. Let $\texttt{Reff}_G(S_1,S_2)$ denote the effective resistance between the vertices $s_1$ and $s_2$ in the graph $H$. The following is a consequence of our main result:

\begin{theorem}\label{thm:res}
In any weighted graph $G$, there are two sets of vertices $S_1$ and $S_2$ for which $\texttt{Reff}_G(S_1,S_2)\ge 1/(25600\lambda_G\min(\text{vol}_G(S_1),\text{vol}_G(S_2)))$. Furthermore, for any pair of sets $S_1',S_2'$, $\texttt{Reff}_G(S_1',S_2')\le 2/(\lambda_G\min(\text{vol}_G(S_1'),\text{vol}_G(S_2'))$.
\end{theorem}

\cite{CRRST97} proved the upper bound present in this result when $|S_1'| = |S_2'| = 1$. We prove Theorem \ref{thm:res} in Appendix \ref{app:res}.

\subsection{Graph Partitioning}

Effective resistance in spectral graph theory has been used several times recently (for example \cite{MST15,AALG18}) to obtain improved graph partitioning results. $1/\lambda_G$ may not yield a good approximation to the effective resistance between pairs of vertices \cite{CRRST97}. For example, on an $n$-vertex grid graph $G$, all effective resistances are between $\Omega(1)$ and $O(\log n)$, but $\lambda_G = \Theta(1/n)$. Theorem \ref{thm:res} closes this gap by considering pairs of \emph{sets} of vertices, not just pairs of vertices.

Cheeger's inequality is the starting point for analysis of spectral partitioning. In some partitioning tasks, cutting the graph does not make sense. For example, spectral partitioning is an important tool in image segmentation \cite{SM00,MS00}. Graph partitioning makes the most sense in image segmentation when one wants to find an object with a sharp boundary. However, in many images, like the one in Figure \ref{fig:pic} on the right, objects may have fuzzy boundaries. In these cases, it is not clear which cut an image segmentation algorithm should return.

\begin{figure}
\begin{center}
\includegraphics[width=0.4\textwidth]{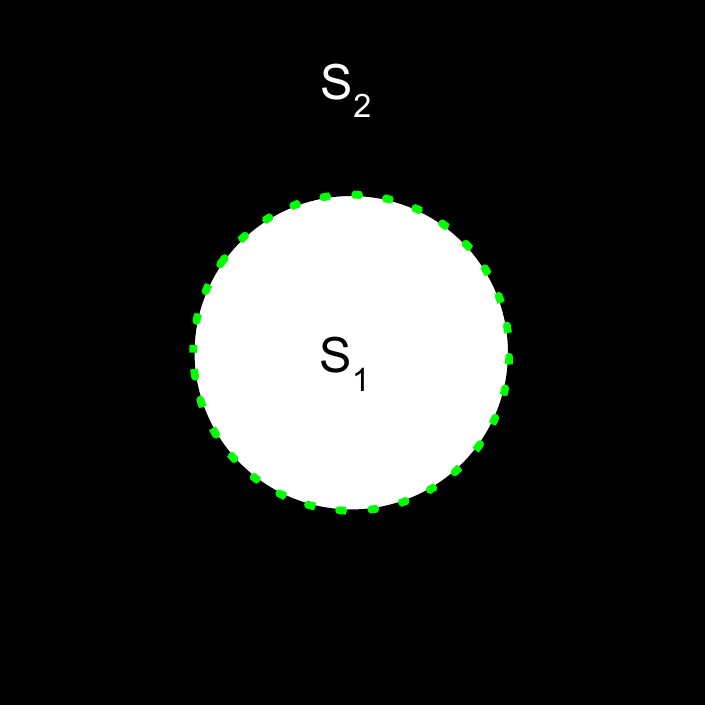}
\includegraphics[width=0.4\textwidth]{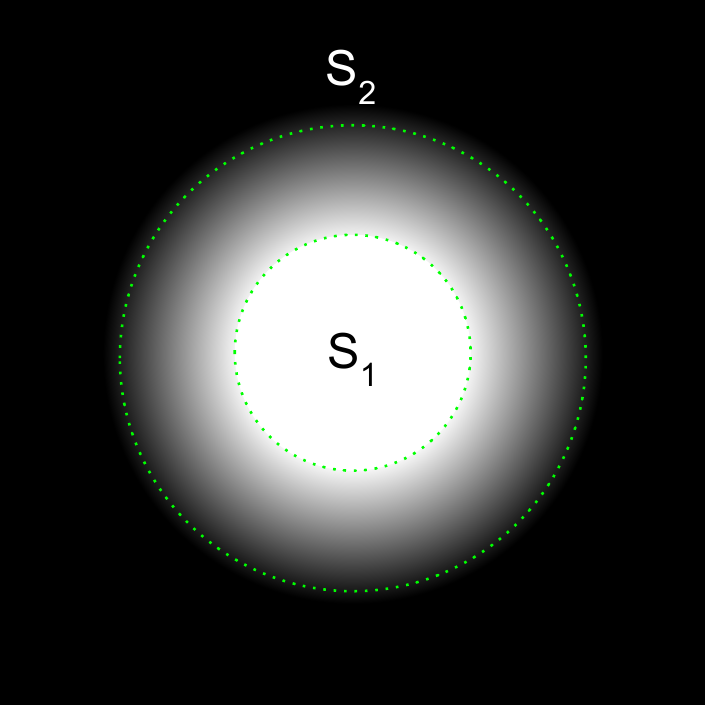}
\end{center}
\caption{Spectral partitioning finds the $S_1$-$S_2$ cut in the left image, but may not in the right due to the presence of many equal weight cuts. The minimum fractional conductance Schur complement cut is displayed in both images.}
\label{fig:pic}
\end{figure}

Considering cuts in Schur complements circumvents this ambiguity. Think of an image as a graph by making a vertex for each pixel and making an edge between adjacent pixels, where the weight on an edge is inversely related to the disparity between the colors of the endpoint pixels for the edge. An optimal segmentation in our setting would consist of the two sets $S_1$ and $S_2$ corresponding to pixels on either side of the fuzzy boundary. Computing the Schur complement of the graph onto $S_1\cup S_2$ eliminates all vertices corresponding to pixels in the boundary.

Some examples in which Cheeger's inequality is not tight illustrate a similar phenomenon in which there are many equally good cuts. For example, let $G$ be an unweighted $n$-vertex cycle. This is a tight example for the upper bound in Cheeger's inequality, as no cut has fractional conductance smaller than $O(1/n)$ despite the fact that $\lambda_G = \Theta(1/n^2)$. Instead, divide the cycle into four equal-sized quarters and let $S_1$ and $S_2$ be two opposing quarters. The Schur complement cut between $S_1$ and $S_2$ has fractional conductance at most $O(1/n^2)$, which matches $\lambda_G$ up to a constant factor.

\begin{figure}
\begin{center}
\includegraphics[width=0.4\textwidth]{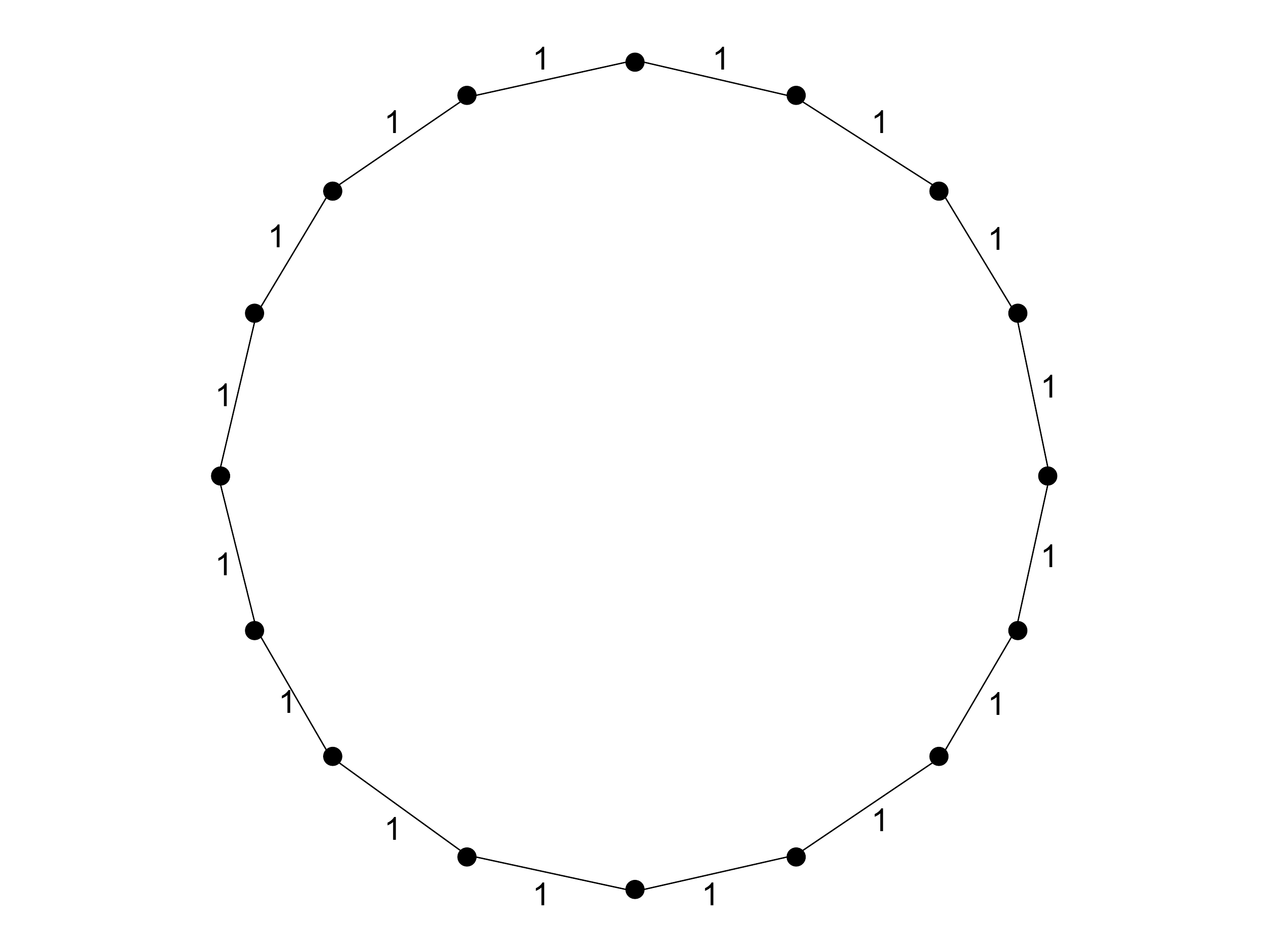}
\includegraphics[width=0.4\textwidth]{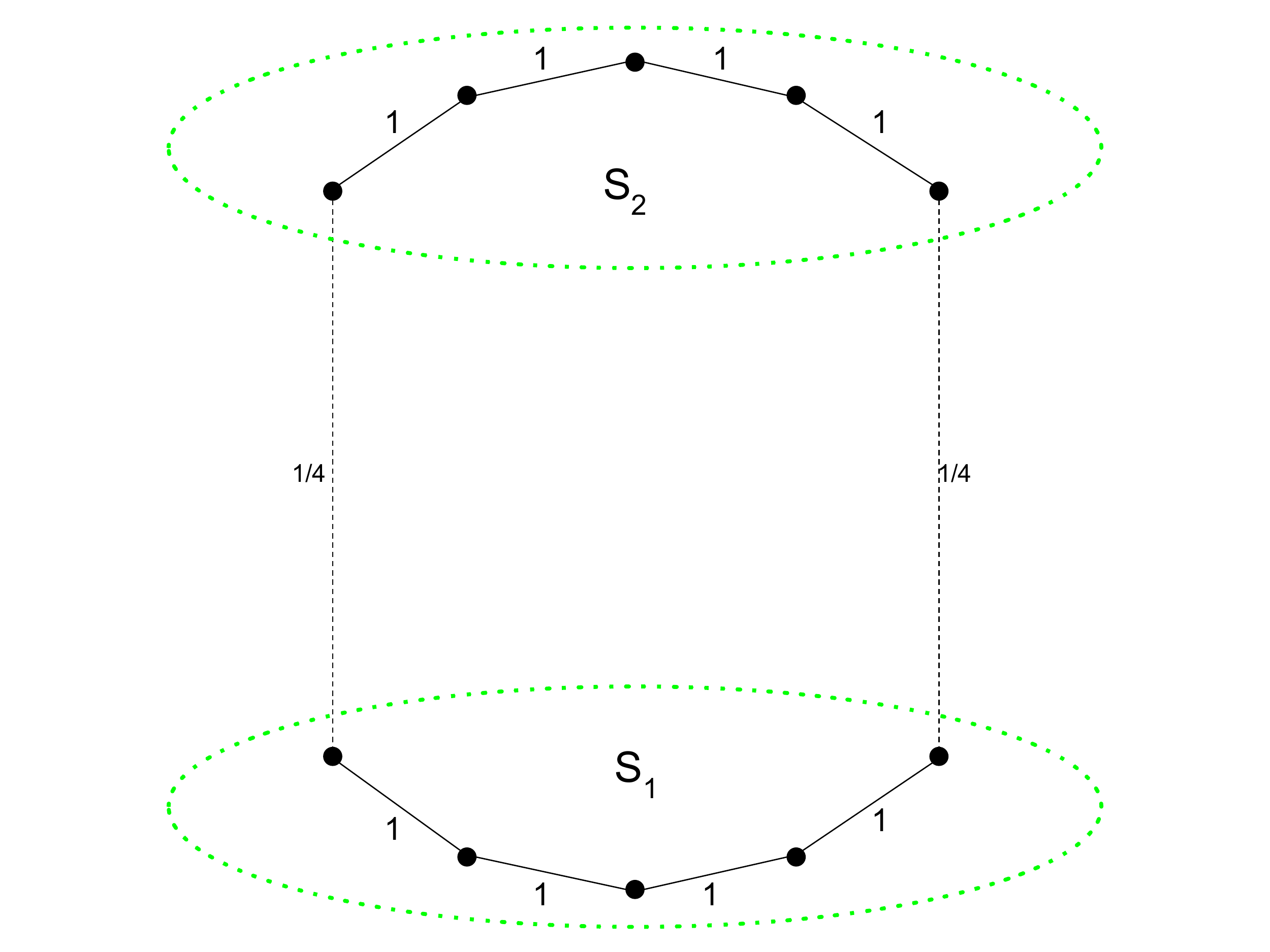}
\caption{A tight example for the upper bound in Cheeger's inequality. The minimum fractional conductance of any cut in this graph is $1/8$, while the fractional conductance of the illustrated Schur complement cut on the right is $2(1/4)/(2(1/4) + 8(1)) = 1/17 < 1/8$.}
\end{center}
\end{figure}

\section{Preliminaries}\label{sec:prelim}

\noindent \textbf{Graph theory}: Consider an undirected, connected graph $H$ with edge weights $\{c_e^H\}_{e\in E(H)}$, $m$ edges, and $n$ vertices. Let $V(H)$ and $E(H)$ denote the vertex and edge sets of $H$ respectively. For two sets of vertices $A,B\subseteq V(H)$, let $E_H(A,B)$ denote the set of edges in $H$ incident with one vertex in $A$ and one vertex in $B$ and let $c^H(A,B) := \sum_{e\in E_H(A,B)} c_e^H$. For a set of edges $F\subseteq E(H)$, let $c^H(F) := \sum_{e\in F} c_e^H$. For a set of vertices $A\subseteq V(H)$, let $\partial_H A := E_H(A,V(H)\setminus A)$. For a vertex $v\in V(H)$, let $\partial_H v := \partial_H \{v\}$ denote the edges incident with $v$ in $H$ and let $c_v^H := \sum_{e \in \partial_H v} c_e^H$. For a set of vertices $A\subseteq V(H)$, let $\text{vol}_H(A) := \sum_{v\in A} c_v^H$. When $A$ and $B$ are disjoint, let $H/(A,B)$ denote the graph with all vertices in $A$ identified to one vertex $a$ and all vertices in $B$ identified to one vertex $b$. Formally, let $H/(A,B)$ be the graph with $V(H/(A,B)) = (V(H)\setminus (A\cup B))\cup \{a,b\}$, embedding $f:V(H)\rightarrow V(H/(A,B))$ with $f(u) := a$ if $u\in A$, $f(u) := b$ if $u\in B$, and $f(u) := u$ otherwise, and edges $\{f(u),f(v)\}$ for all $\{u,v\}\in E(H)$. Let $H/A := H/(A,\emptyset)$.\\

\noindent \textbf{Laplacians}: Let $D_H$ be the $n\times n$ diagonal matrix with rows and columns indexed by vertices in $H$ and $D_H(v,v) = c_v^H$ for all $v\in V(H)$. Let $A_H$ be the adjacency matrix of $H$; that is the matrix with $A_H(u,v) = c_{uv}^H$ for all $u,v\in V(H)$. Let $L_H := D_H - A_H$ be the Laplacian matrix of $H$. Let $N_H := D_H^{-1/2} L_H D_H^{-1/2}$ denote the normalized Laplacian matrix of $H$. For a matrix $M$, let $M^{\dagger}$ denote the Moore-Penrose pseudoinverse of $M$. For subsets $A$ and $B$ of rows and columns of $M$ respectively, let $M[A,B]$ denote the $|A|\times |B|$ submatrix of $M$ restricted to those rows and columns. For a set of vertices $S\in V(H)$, let $\textbf{1}_S$ denote the indicator vector for the set $S$. For two vertices $u,v\in \mathbb{R}^n$, let $\chi_{uv} := \textbf{1}_{\{u\}} - \textbf{1}_{\{u\}}$. When the graph is clear from context, we omit $H$ from all of the subscripts and superscripts of $H$. For a vector $x\in \mathbb{R}^n$, let $x_S\in \mathbb{R}^S$ denote the restriction of $x$ to the coordinates in $S$.

Let $\lambda_H$ denote the smallest nonzero eigenvalue of $N_H$. Equivalently, $$\lambda_H := \min_{x\in \mathbb{R}^n : x^T D_H^{1/2} \textbf{1}_{V(H)} = 0} \frac{x^T N_H x}{x^T x}$$ For any set of vertices $X\subseteq V(H)$, let $$L_{\texttt{Schur}(H,X)} := L_H[X,X] - L_H[X,V(H)\setminus X] L_H[V(H)\setminus X, V(H)\setminus X]^{-1} L_H[V(H)\setminus X, X]$$ where brackets denote submatrices with the indexed rows and columns. The following fact applies specifically to Laplacian matrices:

\begin{remark}[Fact 2.3.6 of \cite{K17}]\label{rmk:schur-graph}
For any graph $H$ and any $X\subseteq V(H)$, $L_{\texttt{Schur}(H,X)}$ is the Laplacian matrix of an undirected graph.
\end{remark}

Let $\texttt{Schur}(H,X)$ denote the graph referred to in Remark \ref{rmk:schur-graph}. Schur complementation commutes with edge contraction and deletion and is associative:

\begin{theorem}[Lemma 4.1 of \cite{CDN89}, statement from \cite{K17}]\label{thm:com-schur}
Given $H$, $S\subseteq V(H)$, and any edge $e$ with both endpoints in $S$,

$$\texttt{Schur}(H\setminus e,S) = \texttt{Schur}(H,S)\setminus e$$
and, for any pair of vertices $x,y\in S$,

$$\texttt{Schur}(H/\{x,y\},S) = \texttt{Schur}(H,S)/\{x,y\}$$
\end{theorem}

\begin{theorem}\label{thm:assoc-schur}
Given $H$ and two sets of vertices $X\subseteq Y\subseteq V(H)$, $\texttt{Schur}(\texttt{Schur}(H,Y),X) = \texttt{Schur}(H,X)$.
\end{theorem}

The following property follows from the definition of Schur complements:

\begin{remark}\label{rmk:schur}
Let $H$ be a graph and $S\subseteq V(H)$. Let $I := \texttt{Schur}(H,S)$. For any $x\in \mathbb{R}^{V(H)}$ that is supported on $S$ with $x^T \textbf{1}_{V(H)} = 0$,

$$x^T L_H^{\dagger} x = x_S^T L_I^{\dagger} x_S$$
\end{remark}

The weight of edges in this graph can be computed using the following folklore fact, which we prove for completeness:

\begin{theorem}\label{thm:schur-cond-1}
For two disjoint sets $C,D\subseteq V(H)$, let $I := \texttt{Schur}(H,C\cup D)$. Then $$c^I(C,D) = \frac{1}{\chi_{cd}^T L_{H/(C,D)}^{\dagger} \chi_{cd}}$$
\end{theorem}

\begin{proof}
By definition, $c^I(C,D) = c^{I/(C,D)}(\{c\},\{d\})$. By Theorem \ref{thm:com-schur}, $I/(C,D) = \texttt{Schur}(H/(C,D),\{c,d\})$. By Remark \ref{rmk:schur}, $c^{\texttt{Schur}(H/(C,D),\{c,d\})}(\{c\},\{d\}) = \frac{1}{\chi_{cd}^T L_{H/(C,D)}^{\dagger} \chi_{cd}}$. Combining these equalities gives the desired result.
\end{proof}

We also use the following folklore fact about electrical flows, which we prove for the sake of completeness:

\begin{theorem}\label{thm:schur-cond-2}
For two vertices $s,t\in V(H)$, $$\chi_{st}^T L_H^{\dagger} \chi_{st} = \frac{1}{\min_{p\in \mathbb{R}^{V(H)}: p_s \le 0, p_t \ge 1} p^T L_H p}$$
\end{theorem}

\begin{proof}
We first show that
$$\chi_{st}^T L_H^{\dagger}\chi_{st} = \frac{1}{\min_{p\in \mathbb{R}^{V(H)}: p_s = 0, p_t = 1} p^TL_Hp}$$
Taking the gradient of the objective $p^TL_Hp$ shows that that the optimal $p$ are the potentials for an electrical flow with flow conservation at all vertices besides $s$ and $t$. Therefore, $p$ is proportional to $L_H^{\dagger}\chi_{st} + \gamma\textbf{1}$ for some $\gamma\in\mathbb{R}$. The constant of proportionality is $\chi_{st}^T L_H^{\dagger}\chi_{st}$ since the $s$-$t$ potential drop in $p$ is 1. Therefore,

\begin{align*}
\min_{p\in \mathbb{R}^{V(H)}: p_s = 0, p_t = 1} p^TL_Hp &= \left(\frac{L_H^{\dagger}\chi_{st}}{\chi_{st}^T L_H^{\dagger} \chi_{st}}\right)^TL_H\left(\frac{L_H^{\dagger}\chi_{st}}{\chi_{st}^T L_H^{\dagger} \chi_{st}}\right)\\
&= \frac{1}{\chi_{st}^T L_H^{\dagger} \chi_{st}}\\
\end{align*}

The desired result follows from the fact that in the optimal $p$, all potentials are between 0 and 1 inclusive.
\end{proof}

\noindent \textbf{Notions of fractional conductance}: For a set of vertices $A\subseteq V(H)$, let $$\phi_A^H := \frac{c^H(\partial_H(A))}{\min(\text{vol}_H(A),\text{vol}_H(V(H)\setminus A))}$$ be the \textit{fractional conductance} of $A$. Let $$\phi_H := \min_{A\subseteq V(H): A\ne \emptyset} \phi_A^H$$ be the \textit{fractional conductance} of $H$.

For two disjoint sets of vertices $A,B\subseteq V(H)$, let $I := \texttt{Schur}(H,A\cup B)$ and $$\rho_{A,B}
^H := \frac{c^I(A,B)}{\min(\text{vol}_I(A),\text{vol}_I(B))}$$ be the \textit{Schur complement fractional conductance} of the pair of sets $(A,B)$. Define the \textit{Schur complement fractional conductance} of the graph $H$ to be $$\rho_H := \min_{A,B\subseteq V(H) : A\cap B = \emptyset,A\ne \emptyset,B\ne \emptyset} \rho_{A,B}^H$$ It will be helpful to deal with the quantities $$\sigma_{A,B}^H := \frac{c^I(A,B)}{\min(\text{vol}_H(A),\text{vol}_H(B))}$$ and $$\sigma_H := \min_{A,B\subseteq V(H) : A\cap B = \emptyset,A\ne \emptyset,B\ne \emptyset} \sigma_{A,B}^H$$ as well, which we call the \textit{mixed fractional conductances} of $(A,B)$ and $H$ respectively.

The following will be useful in relating $\rho_{A,B}^H$ to $\sigma_{A,B}^H$:

\begin{proposition}\label{prop:vol-mon}
For any two sets $X\subseteq Y\subseteq V(H)$, let $I := \texttt{Schur}(H,Y)$. Then, $$\text{vol}_I(X)\le \text{vol}_H(X)$$
\end{proposition}

\begin{proof}
It suffices to show this result when $|X| = 1$ because $\text{vol}$ is a sum of volumes (degrees) of vertices in the set. Furthermore, by Theorem \ref{thm:assoc-schur}, it suffices to show the result when $|Y| = |V(H)| - 1$. Let $v$ be the unique vertex in $H$ outside of $Y$ and let $u$ be the unique vertex in $X$. Then, by definition of the Schur complement,

\begin{align*}
\text{vol}_I(X) &= c_u^I\\
&= \sum_{w\in V(I)} c_{uw}^I\\
&= \sum_{w\in V(I)} \left(c_{uw}^H + \frac{c_{uv}^Hc_{vw}^H}{c_v^H}\right)\\
&= \left(\sum_{w\in V(I)} c_{uw}^H\right) + \frac{c_{uv}^H}{c_v^H}\left(\sum_{w\in V(I)} c_{vw}^H\right)\\
&\le \left(\sum_{w\in V(I)} c_{uw}^H\right) + c_{uv}^H\\
&= c_u^H\\
&= \text{vol}_H(X)\\
\end{align*}
as desired.
\end{proof}

To prove the upper bound, we given an algorithm for constructing a low fractional conductance Schur complement cut. The following result is helpful for making this algorithm take near-linear time:

\begin{theorem}[Theorem 8.2 of \cite{V13}]\label{thm:e-vec}
Given a graph $H$, there is a $\tilde{O}(m)$-time algorithm that produces a vector $x \leftarrow \ApxFiedler(H) \in \mathbb{R}^{V(H)}$ with $x^T D_H^{1/2}\textbf{1}_{V(H)} = 0$ for which

$$x^T N_H x \le 2\lambda_H x^T x$$
\end{theorem}

\section{Lower bound}

We now show the first inequality in Theorem \ref{thm:main}, which follows from the following lemma by Proposition \ref{prop:vol-mon}, which implies that $\sigma_G\le \rho_G$.

\begin{lemma}\label{lem:lower}
$$\lambda_G\le 2\sigma_G$$
\end{lemma}

\begin{proof}
We lower bound the Schur complement fractional conductance of any pair of disjoint sets $A,B\subseteq V(G)$. Let $I := \texttt{Schur}(G,A\cup B)$. Let $P$ be the $(A\cup B) \times (A\cup B)$ diagonal matrix with $P(u,u) = c_u^G$ for each $u\in A\cup B$. We start by lower bounding the minimum nonzero eigenvalue $\lambda$ of the matrix $P^{-1/2} L_I P^{-1/2}$. Let $\lambda_{\max}(M)$ denote the maximum eigenvalue of a symmetric matrix $M$. By definition of the Moore-Penrose pseudoinverse, $$1/\lambda = \lambda_{\max}(P^{1/2} L_I^{\dagger} P^{1/2})$$ By Remark \ref{rmk:schur}, $$\lambda_{\max}(P^{1/2} L_I^{\dagger} P^{1/2}) \le \lambda_{\max}(N_G^{\dagger}) = 1/\lambda_G$$ Therefore, $\lambda \ge \lambda_G$. We now plug in a test vector. Let $$z := P^{1/2}\left(\frac{\textbf{1}_A}{\text{vol}_G(A)} - \frac{\textbf{1}_B}{\text{vol}_G(B)}\right)$$ $z^T (P^{1/2}\textbf{1}_{V(I)}) = 0$, so

\begin{align*}
\lambda_G &\le \lambda\\
&= \min_{x\in \mathbb{R}^{A\cup B}: x^T P^{1/2} \textbf{1}_{V(I)} = 0} \frac{x^T (P^{-1/2} L_I P^{-1/2}) x}{x^T x}\\
&\le \frac{z^T (P^{-1/2} L_I P^{-1/2}) z}{z^T z}\\
&= \frac{c^I(A,B)\left((1/\text{vol}_G(A)) + (1/\text{vol}_G(B))\right)^2}{(\text{vol}_G(A)/\text{vol}_G(A)^2) + (\text{vol}_G(B)/\text{vol}_G(B)^2)}\\
&= \frac{c^I(A,B)\text{vol}_G(A\cup B)}{\text{vol}_G(A)\text{vol}_G(B)}\\
&\le 2\sigma_{A,B}^G\\
\end{align*}
\end{proof}

\section{Upper bound}

We now show the second inequality in Theorem \ref{thm:main}:

\begin{lemma}\label{lem:upper}
$$\rho_G\le 25600\lambda_G$$
\end{lemma}

To prove this lemma, we need to find a pair of sets $A$ and $B$ with low Schur complement fractional conductance:

\begin{lemma}\label{lem:upper-main}
There is a near-linear time algorithm $\SweepCut(G)$ that takes in a graph $G$ with $\lambda_G\le 1/25600$ and outputs a pair of nonempty sets $A$ and $B$ with the following properties:

\begin{itemize}
\item (Low Schur complement fractional conductance) $\sigma_{A,B}^G\le 640\lambda_G$
\item (Large interior) $\phi_A^G \le 1/4$ and $\phi_B^G \le 1/4$
\end{itemize}
\end{lemma}

We now prove Lemma \ref{lem:upper} given Lemma \ref{lem:upper-main}:

\begin{proof}[Proof of Lemma \ref{lem:upper} given Lemma \ref{lem:upper-main}]
Let $I := \texttt{Schur}(G,A\cup B)$. For any two vertices $u,v\in A\cup B$, $c_{uv}^I \ge c_{uv}^G$. Therefore, $\text{vol}_I(A)\ge 2\sum_{u,v\in A} c_{uv}^G$ and $\text{vol}_I(B)\ge 2\sum_{u,v\in B} c_{uv}^G$. By the ``Large interior'' guarantee of Lemma \ref{lem:upper-main}, $2\sum_{u,v\in A} c_{uv}^G \ge (3/4)\text{vol}_G(A)$ and $2\sum_{u,v\in B} c_{uv}^G \ge (3/4)\text{vol}_G(B)$. Therefore, $$\rho_{A,B}^G\le 4/3\sigma_{A,B}^G\le 1280\lambda_G$$ by the ``Low Schur complement fractional conductance'' guarantee when $\lambda_G\le 1/25600$, as desired. When $\lambda_G > 1/25600$, the lemma is trivially true, as desired.
\end{proof}

Now, we implement $\SweepCut$. The standard Cheeger sweep examines all thresholds $q\in \mathbb{R}$ and for each threshold, computes the fractional conductance of the cut $\partial S_{\le q}$ of edges from vertices with eigenvector coordinate at most $q$ to ones greater than $q$. Instead, the algorithm $\SweepCut$ examines all thresholds $q\in \mathbb{R}$ and computes an upper bound (a proxy) for the $\sigma_{S_{\le q/2},S_{\ge q}}^G$ for each positive $q$ and $\sigma_{S_{\le q},S_{\ge q/2}}^G$ for each negative $q$. Let $I_q := \texttt{Schur}(G,S_{\ge q}\cup S_{\le q/2})$ for $q > 0$ and $I_q := \texttt{Schur}(G,S_{\le q}\cup S_{\ge q/2})$. Let $\kappa_q(y) := \min(q,\max(q/2,y))$ for $q > 0$ and $\kappa_q(y) = \min(q/2,\max(q,y))$ for $q\le 0$. The proxy is the following quantity, which is defined for a specific shift of the Rayleigh quotient minimizer $y\in \mathbb{R}^{V(G)}$.

$$\widehat{c}^{I_q}(S_{\ge q},S_{\le q/2}) := \frac{4}{q^2} \sum_{e=uv\in E(G)} c_e^G (\kappa_q(y_u) - \kappa_q(y_v))^2$$
for $q > 0$ and

$$\widehat{c}^{I_q}(S_{\le q},S_{\ge q/2}) := \frac{4}{q^2} \sum_{e=uv\in E(G)} c_e^G (\kappa_q(y_u) - \kappa_q(y_v))^2$$
for $q\le 0$. We now show that this is indeed an upper bound:

\begin{proposition}\label{prop:proxy-bound}
For all $q > 0$,

$$c^{I_q}(S_{\le q/2},S_{\ge q}) \le \widehat{c}^{I_q}(S_{\le q/2},S_{\ge q})$$
For all $q\le 0$,

$$c^{I_q}(S_{\le q},S_{\ge q/2}) \le \widehat{c}^{I_q}(S_{\le q},S_{\ge q/2})$$
\end{proposition}

\begin{proof}
We focus on the $q > 0$, as the reasoning for the $q\le 0$ case is the same. By Theorems \ref{thm:schur-cond-1} and \ref{thm:schur-cond-2},

$$c^{I_q}(S_{\le q/2},S_{\ge q}) = \min_{p\in \mathbb{R}^{V(G)}: p_a \le 0 \forall a\in S_{\le q/2}, p_a \ge 1 \forall a\in S_{\ge q}} p^TL_Gp$$
The vector $p$ with $p_a := \frac{2}{q}\kappa_q(y_a) - 1$ for all vertices $a\in V(G)$ is a feasible solution to the above optimization problem with objective value $\widehat{c}^{I_q}(S_{\le q/2},S_{\ge q})$. This is the desired result.

\end{proof}

This proxy allows us to relate Schur complement fractional conductances together across different thresholds $q$ in a similar proof to the proof of the upper bound of Cheeger's inequality given in \cite{T11}. One complication in our case is that Schur complements for different values of $q$ overlap in their eliminated vertices. Our choice of $\le q/2$, $\ge q$ plays a key role here (as opposed to $\le 0$, $\ge q$, for example) in ensuring that the overlap is small. We now give the algorithm $\SweepCut$:\\

\begin{algorithm}[H]
\SetAlgoLined
\DontPrintSemicolon
\caption{$\SweepCut(G)$}

\KwIn{A graph $G$ with $\lambda_G\le 1/25600$}

\KwOut{Two sets of vertices $A$ and $B$ satisfying the guarantees of Lemma \ref{lem:upper-main}}

$z\gets $ vector with $z^T N_G z\le 2\lambda_G z^T z$ and $z^T (D_G^{1/2} \textbf{1}_{V(G)}) = 0$\;

$x\gets D_G^{-1/2} z$\;

$y\gets x - \alpha \textbf{1}_{V(G)}$ for a value $\alpha$ such that $\text{vol}_G(\{v : y_v \le 0\}) \ge \text{vol}_G(V(G)) / 2$ and  $\text{vol}_G(\{v : y_v \ge 0\}) \ge \text{vol}_G(V(G)) / 2$\;

\ForEach{$q \in \mathbb{R}$}{
    $S_{\ge q}\gets $ vertices with $y_v \ge q$\;

    $S_{\le q}\gets $ vertices with $y_v \le q$\;
} 

\ForEach{$q > 0$}{
    \If{(1) $\widehat{c}^{I_q}(S_{\le q/2},S_{\ge q}) \le 640\lambda_G \min(\text{vol}_G(S_{\le q/2}), \text{vol}_G(S_{\ge q})))$, (2) $c^G(\partial S_{\ge q/2})\le 1/4\text{vol}_G(S_{\ge q})$, and (3) $\phi_{S_{\ge q}}\le 1/4$}{

        \Return{$(S_{\le q/2}, S_{\ge q})$}

    }
}

\ForEach{$q \le 0$}{
    \If{(1) $\widehat{c}^{I_q}(S_{\ge q/2},S_{\le q}) \le 640\lambda_G \min(\text{vol}_G(S_{\ge q/2}), \text{vol}_G(S_{\le q})))$, (2) $c^G(\partial S_{\ge q/2})\le 1/4\text{vol}_G(S_{\le q})$, and (3) $\phi_{S_{\le q}}\le 1/4$}{

        \Return{$(S_{\le q}, S_{\ge q/2})$}

    }
}

\end{algorithm}

Our analysis relies on the following key technical result, which we prove in Appendix \ref{app:tech}:

\begin{proposition}\label{prop:kappa-int}
For any $a,b\in \mathbb{R}$,
$$\int_0^{\infty} \frac{(\kappa_q(a) - \kappa_q(b))^2}{q} dq \le 10 (a - b)^2$$
\end{proposition}

\begin{proof}[Proof of Lemma \ref{lem:upper-main}]
\textbf{Algorithm well-definedness.} We start by showing that $\SweepCut$ returns a pair of sets. Assume, for the sake of contradiction, that $\SweepCut$ does not return a pair of sets. Let $I_q := \texttt{Schur}(G, S_{\ge q}\cup S_{\le q/2})$ for $q > 0$ and $I_q := \texttt{Schur}(G, S_{\le q}\cup S_{\ge q/2})$ for $q \le 0$. By the contradiction assumption, for all $q > 0$, $$\text{vol}_G(S_{\ge q})\le \frac{\widehat{c}^{I_q}(S_{\ge q},S_{\le q/2})}{640\lambda_G} + 4 c^G(\partial S_{\ge q}) + 4 c^G(\partial S_{\le q/2})$$ and for all $q < 0$, $$\text{vol}_G(S_{\le q})\le \frac{\widehat{c}^{I_q}(S_{\le q},S_{\ge q/2})}{640\lambda_G} + 4 c^G(\partial S_{\le q}) + 4 c^G(\partial S_{\ge q/2})$$ Since $\sum_{v\in V(G)} c_v^G x_v = 0$, $$\sum_{v\in V(G)} c_v^G x_v^2 \le \sum_{v\in V(G)} c_v^G y_v^2$$ Now, we bound the positive $y_v$ and negative $y_v$ parts of this sum separately. Negating $y$ shows that it suffices to bound the positive part. Order the vertices in $S_{\ge 0}$ in decreasing order by $y_v$ value. Let $v_i$ be the $i$th vertex in this ordering, let $k := |S_{\ge 0}|$, $y_{k+1} := 0$, $y_i := y_{v_i}$, $c_i := c_{v_i}^G$, and $S_i := \{v_1,v_2,\hdots,v_i\}$ for each integer $i \in [k]$. Then

\begin{align*}
\sum_{v\in S_{\ge 0}} c_v^G y_v^2 &= \sum_{i=1}^k c_i y_i^2\\
&= \sum_{i=1}^k (\text{vol}_G(S_i) - \text{vol}_G(S_{i-1}))y_i^2\\
&= \sum_{i=1}^k \text{vol}_G(S_i)(y_i^2 - y_{i+1}^2)\\
&= 2\int_0^{\infty} \text{vol}_G(S_{\ge q}) q dq\\
\end{align*}

\noindent By our volume upper bound from above,

\begin{align*}
2\int_0^{\infty} \text{vol}_G(S_{\ge q}) q dq &\le 2\int_0^{\infty}\frac{\widehat{c}^{I_q}(S_{\ge q},S_{\le q/2})}{640\lambda_G} q dq + 8\int_0^{\infty} c^G(\partial S_{\ge q}) q dq + 8\int_0^{\infty} c^G(\partial S_{\le q/2}) q dq\\
&= 2\int_0^{\infty}\frac{\widehat{c}^{I_q}(S_{\ge q},S_{\le q/2})}{640\lambda_G} q dq + 8\int_0^{\infty} c^G(\partial S_{\ge q}) q dq + 8\int_0^{\infty} c^G(\partial S_{> q/2}) q dq\\
&= 2\int_0^{\infty}\frac{\widehat{c}^{I_q}(S_{\ge q},S_{\le q/2})}{640\lambda_G} q dq + 40\int_0^{\infty} c^G(\partial S_{\ge q}) q dq\\
\end{align*}

\noindent Substitution and Proposition \ref{prop:kappa-int} show that

\begin{align*}
2\int_0^{\infty} \text{vol}_G(S_{\ge q}) q dq&\le 8\sum_{e=uv\in E(G)} c_e^G\int_0^{\infty}\left(\frac{(\kappa_q(y_u) - \kappa_q(y_v))^2}{640\lambda_G q} + 5\mathbbm{1}_{q\in [y_u,y_v]} q\right)dq\\
&\le 8\sum_{e=uv\in E(G)} c_e^G \left(\frac{10}{640\lambda_G}(y_u - y_v)^2 + 5|y_u^2 - y_v^2|\right)
\end{align*}

\noindent By Cauchy-Schwarz,

\begin{align*}
8\sum_{e=uv\in E(G)} c_e^G \left(\frac{10}{640\lambda_G}(y_u - y_v)^2 + 5|y_u^2 - y_v^2|\right) &\le \frac{1}{8\lambda_G}\sum_{e=uv\in E(G)} c_e^G (y_u - y_v)^2\\
&+ 40 \sqrt{\sum_{e=uv\in E(G)} c_e^G(y_u - y_v)^2}\sqrt{\sum_{e=uv\in E(G)} c_e^G(y_u + y_v)^2}\\
&\le \frac{1}{4}\sum_{v\in V(G)} c_v^G x_v^2\\
&+ 80\sqrt{\lambda_G} \sqrt{\sum_{v\in V(G)} c_v^G x_v^2}\sqrt{\sum_{v\in V(G)} c_v^G y_v^2}\\
\end{align*}

\noindent But since $\sum_{v\in V(G)} c_v^G x_v^2\le \sum_{v\in V(G)} c_v^G y_v^2$ and $\lambda_G < 1/25600$, $$\frac{1}{4}\sum_{v\in V(G)} c_v^G x_v^2 + 80\sqrt{\lambda_G} \sqrt{\sum_{v\in V(G)} c_v^G x_v^2}\sqrt{\sum_{v\in V(G)} c_v^G y_v^2} < \frac{1}{2}\sum_{v\in V(G)} c_v^G y_v^2$$
Negating $y$ shows that $\sum_{v\in S_{\le 0}} c_v^G y_v^2 < 1/2\sum_{v\in V(G)} c_v^G y_v^2$ as well. But these statements cannot both hold; a contradiction. Therefore, $\SweepCut$ must output a pair of sets.\\

\noindent \textbf{Runtime.} Computing $z$ takes $\tilde{O}(m)$ time by Theorem \ref{thm:e-vec}. Therefore, it suffices to show that the foreach loops can each be implemented in $O(m)$ time. This implementation is similar to the $O(m)$-time implementation of the Cheeger sweep.

We focus on the first foreach loop, as the second is the same with $q$ negated. First, note that the functions $\phi_{S_{\ge q}}$, $c^G(\partial S_{\ge q/2})$, and $\text{vol}_G(S_{\ge q})$ of $q$ are piecewise constant, with breakpoints at $q = y_u$ and $q = 2y_u$ for each $u\in V(G)$. Furthermore, these functions can be computed for all values in $O(m)$ time using an $O(m)$-time Cheeger sweep for each function.

Therefore, it suffices to compute the value of $\widehat{c}^{I_q}(S_{\le q/2},S_{\ge q})$ for all $q\ge 0$ that are local minima in $O(m)$ time. Let $h(q) := \widehat{c}^{I_q}(S_{\le q/2},S_{\ge q})$. Notice that the functions $h(q)$ and $h'(q)$ are piecewise quadratic and linear functions of $q$ respectively, with breakpoints at $q = y_u$ and $q = 2y_u$. Using five $O(m)$-time Cheeger sweeps, one can compute the $q^2,q$ and $1$ coefficients of $h(q)$ and the $q$ and $1$ coefficients of $h'(q)$ between all pairs of consecutive breakpoints. After computing these coefficents, one can compute the value of each function at a point $q$ in $O(1)$ time. Furthermore, given two consecutive breakpoints $a$ and $b$, one can find all points $q\in (a,b)$ with $h'(q) = 0$ in $O(1)$ time. Each local minimum for $h$ is either a breakpoint or a point with $h'(q) = 0$. Since $h$ and $h'$ have $O(n)$ breakpoints, all local minima can be computed in $O(n)$ time. $h$ can be evaluated at all of these points in $O(n)$ time. Therefore, all local minima of $h$ can be computed in $O(m)$ time. Since the algorithm does return a $q$, some local minimum for $h$ also suffices, so this implementation produces the desired result in $O(m)$ time.\\

\noindent \textbf{Low Schur complement fractional conductance.} By Proposition \ref{prop:proxy-bound}, $$c^{I_q}(S_{\ge q},S_{\le q/2})\le \widehat{c}^{I_q}(S_{\ge q},S_{\le q/2})$$ Therefore, $c^{I_q}(S_{\ge q},S_{\le q/2}) \le 640\lambda_G \min(\text{vol}_G(S_{\ge q}), \text{vol}_G(S_{\le q/2}))$ for $q \ge 0$ by the foreach loop if condition. Repeating this reasoning for $q < 0$ yields the desired result.\\

\noindent \textbf{Large interior.} By definition of $\alpha$, $\text{vol}_G(S_{\ge q}) \le \text{vol}_G(S_{\le q/2})$ for $q > 0$. Since $c^G(\partial S_{\le q/2}) \le 1/4\text{vol}_G(S_{\ge q})$, $\phi_{S_{\le q/2}}\le 1/4$, as desired.

\end{proof}

\noindent \textbf{Acknowledgements}: I want to thank Satish Rao, Nikhil Srivastava, and Rasmus Kyng for helpful discussions.

\bibliographystyle{alpha}
\bibliography{rst}

\appendix

\section{Proof of Theorem \ref{thm:res}}\label{app:res}

\begin{proof}[Proof of Theorem \ref{thm:res}]
For any two sets of vertices $S_1,S_2$ in a graph $G$,

$$\texttt{Reff}_G(S_1,S_2)\min(\text{vol}_G(S_1),\text{vol}_G(S_2)) = \frac{1}{\sigma_{S_1,S_2}^G}$$
Therefore, the desired result follows from Lemmas \ref{lem:upper} and \ref{lem:lower}.

\end{proof}

\section{Proof of Proposition \ref{prop:kappa-int}}\label{app:tech}

\begin{proof}[Proof of Proposition \ref{prop:kappa-int}]
Without loss of generality, suppose that $a\le b$. We break the analysis up into cases:\\

\noindent \textbf{Case 1: $a\le 0$}. In this case, $\kappa_q(a) = q/2$ for all $q\ge 0$, so
\begin{align*}
\int_0^{\infty} \frac{(\kappa_q(a) - \kappa_q(b))^2}{q}dq &= \int_0^b \frac{(q/2 - q)^2}{q}\\
&+ \int_b^{2b}\frac{(q/2 - b)^2}{q}dq\\
&+ \int_{2b}^{\infty}\frac{(q/2 - q/2)^2}{q}dq\\
&= \frac{b^2}{8} + \int_b^{2b} (q/4 - b + b^2/q)dq\\
&= \frac{b^2}{2} - b^2 + b^2(\ln 2)\\
&\le 10(a - b)^2\\ 
\end{align*}
as desired.\\

\noindent \textbf{Case 2: $a > 0$ and $b\le 2a$}. In this case,
\begin{align*}
\int_0^{\infty} \frac{(\kappa_q(a) - \kappa_q(b))^2}{q}dq &= \int_0^a \frac{(q - q)^2}{q}dq\\
&+ \int_a^b \frac{(a - q)^2}{q}dq\\
&+ \int_b^{2a} \frac{(a - b)^2}{q}dq\\
&+ \int_{2a}^{2b} \frac{(q/2 - b)^2}{q}dq\\
&+ \int_{2b}^{\infty} \frac{(q/2 - q/2)^2}{q}dq\\
&\le \int_a^{2b} \frac{(a - b)^2}{q}dq\\
&= (a - b)^2 \ln (2b/a)\\
&\le (a - b)^2 \ln 4 \le 10 (a - b)^2\\
\end{align*}
as desired.\\

\noindent \textbf{Case 3: $a > 0$ and $b > 2a$}. In this case,
\begin{align*}
\int_0^{\infty} \frac{(\kappa_q(a) - \kappa_q(b))^2}{q}dq &= \int_0^a \frac{(q - q)^2}{q}dq\\
&+ \int_a^{2a} \frac{(a - q)^2}{q}dq\\
&+ \int_{2a}^b \frac{(q/2 - q)^2}{q}dq\\
&+ \int_b^{2b} \frac{(q/2 - b)^2}{q}dq\\
&+ \int_{2b}^{\infty} \frac{(q/2 - q/2)^2}{q}dq\\
&\le \int_a^{2b} \frac{(q/2 - q)^2}{q}dq\\
&\le b^2/2\\
&\le 2(a - b)^2 \le 10(a - b)^2\\
\end{align*}
as desired.

\end{proof}

\end{document}